\documentclass[11pt,onecolumn,twoside]{IEEEtran} 
\usepackage{graphicx}
\usepackage{amsmath}
\usepackage{mathrsfs}
\usepackage{mathtools}
\usepackage{amssymb}
\usepackage{float}
\usepackage{url}
\usepackage{verbatim}
\usepackage{amsthm}
\usepackage{enumerate}
\usepackage{layout}
\usepackage{bm}

\newtheorem{theorem}{Theorem}
\newtheorem{lemma}{Lemma}

\newtheorem{corollary}{Corollary}
\newtheorem{definition}{Definition}
\theoremstyle{definition}

\newcommand{\naturals}{\ensuremath{\mathbb{N}}}
\newcommand{\Reals}{\ensuremath{\mathbb{R}}}
\newcommand{\expectation}{\ensuremath{\mathbb{E}}}
\newcommand{\Var}{\mathrm{Var}}

\newcommand{\set}{\ensuremath{\mathcal}}

\newcommand{\dif}{\mathrm{d}}
\DeclareMathOperator*{\esssup}{ess\,sup}

\begin{document}
\thispagestyle{empty}
\setcounter{page}{1}
\setlength{\baselineskip}{1.15\baselineskip}

\title{\huge{Tight Lower Bounds for $\alpha$-Divergences Under Moment Constraints and Relations Between Different $\alpha$}\\[0.2cm]}
\author{Tomohiro Nishiyama\\ Email: htam0ybboh@gmail.com}
\date{}
\maketitle
\thispagestyle{empty}

\begin{abstract}
The $\alpha$-divergences include the well-known Kullback-Leibler divergence, Hellinger distance and $\chi^2$-divergence.
In this paper, we derive differential and integral relations between the $\alpha$-divergences that are generalizations of the relation between the Kullback-Leibler divergence and the $\chi^2$-divergence. We also show tight lower bounds for the $\alpha$-divergences under given means and variances. In particular, we show a necessary and sufficient condition such that the binary divergences, which are divergences between probability measures on the same $2$-point set, always attain lower bounds. Kullback-Leibler divergence, Hellinger distance, and $\chi^2$-divergence satisfy this condition.
\end{abstract}
 
\section{Introduction}
The Kullback–Leibler divergence~\cite{kullback1951information} (also known as the relative entropy) and the Hellinger distance~\cite{hellinger1909neue} are divergence measures which play a key role in information theory, statistics, machine learning, physics, signal processing, and other theoretical and applied branches of mathematics.
They both belong to an important class of divergence measures, defined by means of convex functions $f$, and named $f$ -divergences~\cite{csiszar1967information,csiszar1967topological,csiszar1972class}.
The most notable class of $f$-divergence is the $\alpha$-divergence~\cite{amari2016information, cichocki2010families}. By choosing different $\alpha$, we get a large number of well-known divergences as special cases, including Kullback-Leibler divergence, Hellinger distance, and $\chi^2$-divergence~\cite{pearson1900x}.

In this paper, we study relations between the $\alpha$-divergences for different $\alpha$, and derive the tight lower bounds for the $\alpha$-divergences under given means and variances.
The relation between the Kullback-Leibler divergence and the $\chi^2$-divergence was shown in ~\cite{8919677,audenaert2014quantum, nishiyama2020relations}, and we generalize this relation for general $\alpha$ and $\alpha+1$. Regarding the lower bounds under given means and variances, there are some works for the $\chi^2$-divergence and the Hellinger distance~\cite{chapman1951minimum, dashti2013bayesian, katsoulakis2017scalable}. Recently, for the Kullback-Leibler divergence and the Hellinger distance, we showed that the tight lower bounds are all attained by their binary divergences that are divergences between probability measures on the same $2$-point set~\cite{nishiyama2020relations, nishiyama2020tight}.
Our motivation is to study necessary and sufficient conditions for $\alpha$ such that the binary $\alpha$-divergences always attain lower bounds. Furthermore, we show tight lower bounds under given means and variances for the R\'{e}nyi divergences~\cite{renyi1961measures}, which are closely related to the $\alpha$-divergences.

In this work, Section 2 presents notation and definitions, Section 3 refers to the main results, Section 4 shows the proofs of the main results. 
Finally, Section 5 concludes this paper, and lemmas that are necessary for the proofs of the main results are proved in Appendices.

\section{Preliminaries}

This section provides definitions of divergence measures which are used in this paper.
\begin{definition} {\rm \label{def:fD} \cite[p.~4398]{liese2006divergences}}
Let $P$ and $Q$ be probability measures, let $\mu$ be a dominating measure
of $P$ and $Q$ (i.e., $P, Q \ll \mu$), and let $p := \frac{\mathrm{d}P}{\mathrm{d}\mu}$
and $q := \frac{\mathrm{d}Q}{\mathrm{d}\mu}$ be the densities of $P$ and $Q$ with respect
to $\mu$. The {\em $f$-divergence} from $P$ to $Q$ is given by
\begin{align} \label{eq:fD}
D_f(P\|Q) := \int q \, f \Bigl(\frac{p}{q}\Bigr) \, \mathrm{d}\mu,
\end{align}
where
\begin{align}
& f(0) := \underset{t \to 0^+}{\lim} \, f(t), \quad  0 f\biggl(\frac{0}{0}\biggr) := 0, \nonumber \\[0.1cm] 
& 0 f\biggl(\frac{a}{0}\biggr) \nonumber
:= \lim_{t \to 0^+} \, t f\biggl(\frac{a}{t}\biggr)
= a \lim_{u \to \infty} \frac{f(u)}{u}, \quad a>0.
\end{align}
It should be noted that the right side of \eqref{eq:fD} is invariant
in the dominating measure $\mu$.
\end{definition}
\begin{definition}  {\rm \label{def:alpha-div} \cite{cichocki2010families}} 
The basic {\em asymmetric alpha-divergence} is the $f$-divergence with 
\begin{align}
f(t) := 
\begin{dcases} 
\frac{t^\alpha-t}{\alpha(\alpha-1)}, & \alpha\neq 0,1, \nonumber \\
-\log t, & \alpha=0, \nonumber \\  
t \log t, & \alpha=1. \nonumber 
\end{dcases}
\end{align}
for $t >0$,
\begin{align}
\label{eq_alpha_div}
D_A^{(\alpha)}(P\|Q) := 
\begin{dcases}
\frac{1}{\alpha(\alpha-1)}\Bigl(\int p^\alpha q^{1-\alpha} \mathrm{d}\mu -1\Bigr), & \alpha\neq 0,1, \\ 
\int q\log\frac{q}{p} \mathrm{d}\mu=:D(Q\|P), & \alpha=0, \\
\int p\log\frac{p}{q} \mathrm{d}\mu=:D(P\|Q), & \alpha=1,
\end{dcases}
\end{align}
where $D(P\|Q)$ denotes the Kullback-Leibler divergence (relative entropy).
\end{definition}
In the special case for $\alpha=2,\; 0.5,\; -1$, we obtain from \eqref{eq_alpha_div} the well known Pearson Chi-square,
Hellinger and Neyman Chi-square distances, given respectively by
\begin{align}
D_A^{(2)}(P\|Q)=\frac{1}{2}\chi^2_P(P\|Q)=\frac{1}{2} \int \frac{(p-q)^2}{q} \mathrm{d}\mu,  \nonumber\\ 
D_A^{(1/2)}(P\|Q)=4H^2(P,Q)=2\int (\sqrt{p}-\sqrt{q})^2 \mathrm{d}\mu, \nonumber \\  
D_A^{(-1)}(P\|Q)=\frac12\chi^2_N(P\|Q)=\frac12 \int \frac{(p-q)^2}{p} \mathrm{d}\mu.   \nonumber
\end{align}

The $\alpha$-divergences have duality as follows. \\
\textbf{Duality:} \\
\begin{align}
\label{eq-duality}
D_A^{(\alpha)}(P\|Q)=D_A^{(1-\alpha)}(Q\|P).
\end{align}
The R\'{e}nyi divergences~\cite{renyi1961measures} are closely related to $\alpha$-divergences. 
\begin{definition}
\label{def:renyi-div}
The {\em R\'{e}nyi divergence} for the simple orders $\alpha\in (0,1)\cup (1,\infty)$ is defined as
\begin{align}
D_R^{(\alpha)}(P\|Q):= \frac{1}{\alpha-1}\log\int p^\alpha q^{1-\alpha} \mathrm{d}\mu=\frac{1}{\alpha-1}\log\Bigl(1+\alpha(\alpha-1)D_A^{(\alpha)}(P\|Q)\Bigr). \nonumber
\end{align}
For the extended orders, the {\em R\'{e}nyi divergence} is defined as
\begin{align}
D_R^{(\alpha)}(P\|Q) := 
\begin{dcases}
-\log Q(p>0), & \alpha=0, \nonumber \\ 
D(P\|Q)=D_A^{(1)}(P\|Q), & \alpha=1, \\
\log\esssup \frac{p(Z)}{q(Z)} , & \alpha=\infty, 
\end{dcases}
\end{align}
where $Z\sim \mu$.
\end{definition}
\begin{definition} \label{def: probability_set}
Let us define a set of pairs of probability measures $(P,Q)$ defined on $n$-point set $\{u_1, u_2, \cdots, u_n\}$ by $\set{P}_n$, where $\{u_i\}_{1\leq i\leq n}$ are arbitrary real numbers.
\end{definition}
If $m< n$, $\set{P}_m$ is a subset of $\set{P}_n$.
\begin{definition} 
Let $P$ and $Q$ be probability measures defined on a measurable space $(\Reals, \mathscr{B})$, where $\Reals$ is the real line and $\mathscr{B}$ is the Borel $\sigma$-algebra of subsets of 
$\Reals$.
Let $\set{P}[m_P, \sigma_P, m_Q, \sigma_Q]$ be a set of pairs of probability measures $(P,Q)$ under given means and variances, i.e.,
\begin{align}
\label{constraints}
& \expectation[X] =: m_P, \; \expectation[Y] =: m_Q,
\quad \Var(X) =: \sigma_P^2, \;  \Var(Y) =: \sigma_Q^2,
\end{align}
where $X\sim P$ and $Y\sim Q$.
\end{definition}

\begin{definition} \label{def:binary alpha div}
The {\em binary $\alpha$-divergence} is defined for $(R,S)\in \set{P}_2$.
\begin{align}
\label{eq-binary alpha div}
d_A^{(\alpha)}(r\|s) :=
\begin{dcases}
\frac{1}{\alpha(\alpha-1)} \Bigl(r^\alpha s^{1-\alpha}+(1-r)^\alpha (1-s)^{1-\alpha}-1\Bigr), & \alpha\neq 0,1, \\ 
s\log\frac{s}{r}+(1-s)\log\frac{1-s}{1-r}, & \alpha=0, \\
r\log\frac{r}{s}+(1-r)\log\frac{1-r}{1-s}, & \alpha=1. 
\end{dcases}
\end{align}
where $R(u_1) = r$ and $S(u_1) = s$.
\end{definition}

\begin{definition} 
The {\em binary R\'{e}nyi divergence} for orders $\alpha\in [0,\infty]$ is defined for $(R,S)\in \set{P}_2$.
\begin{align}
d_R^{(\alpha)}(r\|s) := 
\begin{dcases}
\frac{1}{\alpha-1}\log\Bigl(r^\alpha s^{1-\alpha} + (1-r)^\alpha(1-s)^{1-\alpha}\Bigr), & \alpha\neq 0,1, \nonumber \\
-\log\Bigl(s1\{r> 0\} + (1-s)1\{1-r> 0\}\Bigr), & \alpha=0, \nonumber \\ 
r\log\frac{r}{s}+(1-r)\log\frac{1-r}{1-s}, & \alpha=1, \nonumber \\
\log\max\Bigl(\frac{r}{s}, \frac{1-r}{1-s}\Bigr), & \alpha=\infty.
\end{dcases}
\end{align}
where $1\{\mbox{relation}\}$ denotes the indicator function. 
\end{definition}

\section{Main results}
\begin{theorem} \label{theorem_relation_alpha}
Let $Q_t:= (Q-P)t+P$ and $P_t:= (P-Q)t+Q=Q_{1-t}$ for $t\in[0,1]$.
Then, for all $t\in(0,1]$, 
\begin{align}
\label{relation_alpha_1}
D_A^{(\alpha+1)}(P\|Q_t)=\frac{t^{2-\alpha}}{(\alpha+1)}\frac{\dif}{\dif t}\Bigl(t^{\alpha-1}D_A^{(\alpha)}(P\|Q_t)\Bigr), \quad \alpha\neq -1, \\[0.1cm]
\label{relation_alpha_2}
D_A^{(\alpha)}(P_t\|Q)=\frac{t^{2+\alpha}}{(-\alpha+1)}\frac{\dif}{\dif t}\Bigl(t^{-\alpha-1}D_A^{(\alpha+1)}(P_t\|Q)\Bigr), \quad \alpha\neq 1. 
\end{align}
\end{theorem}

\begin{proof}
See Section~\ref{section: proofs}.
\end{proof}

\begin{corollary} \label{cor_relation_alpha}
For all $t\in(0,1]$, 
\begin{align}
\label{relation_alpha_3}
D_A^{(\alpha)}(P\|Q_t)=(\alpha+1)t^{1-\alpha}\int_0^t s^{\alpha-2}D_A^{(\alpha+1)}(P\|Q_s)\mathrm{d}s, \quad \alpha> -1, \\[0.1cm]
\label{relation_alpha_4}
D_A^{(\alpha+1)}(P_t\|Q)=(-\alpha+1)t^{1+\alpha}\int_0^t s^{-\alpha-2}D_A^{(\alpha)}(P_s\|Q)\mathrm{d}s, \quad \alpha< 1. 
\end{align}
\end{corollary}
The relation \eqref{relation_alpha_3} for $\alpha =1$ yields the relation between the Kullback-Leibler divergence and the $\chi^2$-divergence.
\begin{align}
D(P\|Q_t)=\int_0^t s^{-1}\chi_P^2(P\|Q_s)\mathrm{d}s.
\end{align}

\begin{proof}
By the Taylor expansion for a differentiable function $f(\cdot)$ such that $f(1)=0$, for efficiently small $s$, we obtain
\begin{align}
\label{cor_1}
\int q_sf\Bigl(\frac{p}{q_s}\Bigr)\mathrm{d}\mu=\int q_s \Bigl(f'(1)\Bigl(\frac{p}{q_s}-1\Bigr)+\frac{f''(1)}{2}{\Bigl(\frac{p}{q_s}-1\Bigr)}^2 +O(s^3)\Bigr) \mathrm{d}\mu = \frac{f''(1)s^2}{2} \chi^2_P(P\|Q)+O(s^3).
\end{align}
Since the $\alpha$-divergences belong to $f$-divergences, it follows that
\begin{align}
\label{cor_2}
t^{\alpha-1}D_A^{(\alpha)}(P\|Q_t)=O(t^{\alpha+1}).
\end{align}
Replacing $t$ by $s$, multiplying $(\alpha + 1)s^{\alpha-2}$ and integrating both sides of \eqref{relation_alpha_1}, we obtain \eqref{relation_alpha_3}. The condition $\alpha>-1$ is due to \eqref{cor_2}. The equality \eqref{relation_alpha_4} follows in a similar way.
 \end{proof}
  
\begin{theorem} \label{theorem_lower_bound}
Let $(P, Q)\in \set{P}[m_P, \sigma_P, m_Q, \sigma_Q]$.
\begin{enumerate}[(a)]
\item
If $m_P\neq m_Q$, the binary $\alpha$-divergence always attains a lower bound under given means and variances if and only if $\alpha\in[-1,2]$.
\begin{align} \label{lb_alpha}
D_A^{(\alpha)}(P\|Q) \geq d_A^{(\alpha)}(r\|s), 
\end{align}
where 
\begin{align}
\label{r}
& r := \frac12 + \frac{\sigma_Q^2-\sigma_P^2+a^2}{4av} \in [0,1], \\
\label{s}
& s :=  \frac12 + \frac{\sigma_Q^2-\sigma_P^2-a^2}{4av}\in [0,1], \\
\label{a}
& a:= m_P - m_Q, \\
\label{v}
& v:= \frac{1}{2|a|}\sqrt{(\sigma_Q^2-\sigma_P^2)^2 + 2a^2(\sigma_P^2+\sigma_Q^2)+a^4}.
\end{align}

\item
The lower bound in the right side of \eqref{lb_alpha} is attained for $(R,S)\in \set{P}_2\cap \set{P}[m_P, \sigma_P, m_Q, \sigma_Q]$ defined on $\{u_1, u_2\}$, and
\begin{align} \label{eq_binary_prob}
R(u_1) = r, \quad S(u_1) = s,
\end{align}
with $r$ and $s$ in \eqref{r} and \eqref{s}, respectively, and
\begin{align} \label{vec_u_1,2}
& u_1 :=  m_P + \sqrt{\frac{(1-r) \sigma_P^2}{r}},
\quad u_2 := m_P - \sqrt{\frac{r \sigma_P^2}{1-r}}.
\end{align}

\item
If $m_P=m_Q$, then,
\begin{align} \label{mean_equal}
\inf_{(P,Q)\in\set{P}[m_P,\sigma_P, m_Q, \sigma_Q]} D_A^{(\alpha)}(P\|Q)=0.
\end{align}
\end{enumerate}
\end{theorem}

\begin{proof}
See Section~\ref{section: proofs}.
\end{proof}
For the Hellinger distance and the $\chi^2$-divergence, their binary divergences are simplified as 
\begin{align}
d_A^{(2)}(r\|s)&=\frac{a^2}{2\sigma_Q^2}, \nonumber \\
d_A^{(1/2)}(r\|s)&=4\Bigl(1-\sqrt{\frac{(\sigma_P+\sigma_Q)^2}{a^2+(\sigma_P+\sigma_Q)^2}}\Bigr). \nonumber 
\end{align}
See Subsection~\ref{sub_sec:minimum_condition} and~\cite{nishiyama2020tight}[Lemma 2], respectively.

\begin{corollary}
If $m_P\neq m_Q$ and $\alpha\in[0,2]$, the binary R\'{e}nyi divergence always attains a lower bound under given means and variances.
\end{corollary}

\begin{proof}
For $\alpha\in(0,2]$, the result follows by Definition~\ref{def:renyi-div} and Theorem~\ref{theorem_lower_bound}.
Since the binary R\'{e}nyi divergence is equal to $0$ for $\alpha=0$ and $\sigma_P>0$, we obtain the result.
For $\alpha=\sigma_P=0$, we have $P(u)=1$ at $u=m_P$. Letting $q=Q(u)$, we obtain 
\begin{align}
\label{cor_m_q}
\sum_{u_i\neq m_P} Q(u_i)u_i=m_Q-qm_P, \\[0.1cm]
\label{cor_sigma_q}
\sum_{u_i\neq m_P} Q(u_i)u_i^2=\sigma_Q^2 + m_Q^2-qm_P^2.
\end{align}
By the Cauchy-Schwarz inequality, it follows that
\begin{align}
\label{cor_inequality}
\Bigl(\sum_{u_i\neq m_P} Q(u_i)u_i\Bigr)^2=\Bigl(\sum_{u_i\neq m_P} \sqrt{Q(u_i)}(\sqrt{Q(u_i)}u_i)\Bigr)^2\leq \Bigl(\sum_{u_i\neq m_P} Q(u_i)\Bigr)\Bigl(\sum_{u_i\neq m_P} Q(u_i)u_i^2\Bigr)=(1-q)(\sigma_Q^2 + m_Q^2-qm_P^2).
\end{align}
By combining this inequality with \eqref{cor_m_q}, we obtain
\begin{align}
\label{upper_bound}
q\leq \frac{\sigma_Q^2}{\sigma_Q^2+a^2}.
\end{align}
From \eqref{r} and \eqref{s}, we have $r=1$ and $s=\frac{\sigma_Q^2}{\sigma_Q^2+a^2}$ for $a>0$. By combining \eqref{upper_bound} with Definition \eqref{def:renyi-div}, the result follows. The case $a<0$ can be justified in a similar way.
\end{proof}

\section{Proofs of main results}
\label{section: proofs}
\subsection{Proof of Theorem~\ref{theorem_relation_alpha}}
\begin{proof}
Let $q_t:= (q-p)t+p$ for all $t\in(0,1]$. For $\alpha \neq 0, \;\pm 1$, we obtain
\begin{align}
\label{ra_1}
\frac{\dif}{\dif t}\Bigl(\Bigl(\int p^\alpha q_t^{1-\alpha} \mathrm{d}\mu -1\Bigr)t^{\alpha-1}\Bigr)&=(\alpha-1)t^{\alpha-2}\Bigl(\int\Bigl(-p^\alpha q_t^{-\alpha}(q-p)t + p^\alpha q_t^{1-\alpha}\Bigr) \mathrm{d}\mu-1\Bigr) \\[0.1cm] \nonumber
&=(\alpha-1)t^{\alpha-2}\Bigl(\int\Bigl(p^\alpha q_t^{-\alpha}(-(q-p)t + q_t)\Bigr) \mathrm{d}\mu-1\Bigr) \\[0.1cm] 
&=(\alpha-1)t^{\alpha-2}\Bigl(\int p^{1+\alpha}q_t^{-\alpha}\mathrm{d}\mu-1\Bigr)=(\alpha+1)\alpha(\alpha-1) t^{\alpha-2}D_A^{(\alpha+1)}(P\|Q_t). \nonumber
\end{align}
Dividing $(\alpha+1)\alpha(\alpha-1)t^{\alpha-2}$ both sides of \eqref{ra_1}, we obtain \eqref{relation_alpha_1}.
For $\alpha=1$, we obtain
\begin{align}
\frac{t}{2} \frac{\dif}{\dif t}\int p\log\frac{p}{q_t}\mathrm{d}\mu=\frac{1}{2} \int \frac{p(p-q_t)}{q_t}\mathrm{d}\mu=\frac{1}{2} \int \frac{(q_t-p)^2}{q_t}\mathrm{d}\mu
=D_A^{(2)}(P\|Q_t).
\end{align}
For $\alpha=0$, we obtain
\begin{align}
t^2 \frac{\dif}{\dif t}\int t^{-1}q_t\log\frac{q_t}{p}\mathrm{d}\mu=t^2\int \Bigl(-t^{-2}p\log\frac{q_t}{p}+t^{-1}(q-p) \mathrm{d}\mu\Bigr)=\int p\log\frac{p}{q_t} \mathrm{d}\mu
=D_A^{(1)}(P\|Q_t).
\end{align}
By the duality \eqref{eq-duality}, and swapping $P$ and $Q$ for \eqref{relation_alpha_1}, it follows that
\begin{align}
D_A^{(-\alpha)}(P_t\|Q)=\frac{t^{2-\alpha}}{(\alpha+1)}\frac{\dif}{\dif t}\Bigl(t^{\alpha-1}D_A^{(1-\alpha)}(P_t\|Q)\Bigr).
\end{align}
Replacing $\alpha$ by $-\alpha$ yields \eqref{relation_alpha_2}.
\end{proof}

\subsection{Proof of Theorem~\ref{theorem_lower_bound}}
\label{sub_sec:minimum_condition}
\begin{lemma} \label{minimum_condition}
\label{lem_minimum}
For $R>0$, let $\set{P}_{n, R}\subseteq \set{P}_{n}$ be a set of pairs of probability measures such that $|u_i|\leq R$ for all $i=1,2,\cdots, n$.
Let $(P,Q)\in \set{P}_{n, R}\cap \set{P}[m_P, \sigma_P, m_Q, \sigma_Q]$. If $\alpha\in(0,1)$ and $m_P\neq m_Q$, the global minimum points $(P^*, Q^*, \bm{u}^*)=\mathrm{argmin}_{(P,Q)\in\set{P}_{n,R}\cap \set{P}[m_P, \sigma_P, m_Q, \sigma_Q]} D_A^{(\alpha)}(P\|Q)$ satisfy any of the following conditions.
\begin{enumerate}[(1)]
\item
$(P^*,Q^*)\in\set{P}_2$, and $\max_i |u^*_i| < R$.
\item
$\max_i |u^*_i| = R$.
\end{enumerate}
\end{lemma}

\begin{proof}
See Appendix~\ref{ap_minimum_condition}.
\end{proof}
\begin{lemma} \label{decreasing}
If $\alpha\in(0,1)$ and $m_P\neq m_Q$, the binary $\alpha$-divergence is monotonically decreasing with respect to both $\sigma_P$ and $\sigma_Q$.
\end{lemma}

\begin{proof}
See Appendix~\ref{ap_decreasing}.
\end{proof}

\begin{lemma}
\label{lem_p2}
If $m_P\neq m_Q$, a set $\set{P}_2 \cap \set{P}[m_P, \sigma_P, m_Q, \sigma_Q]$ has one component $(R,S)$ that is given by \eqref{eq_binary_prob}, and \eqref{vec_u_1,2}.
\end{lemma}
\begin{proof}
See~\cite{nishiyama2020tight}[Lemma 1].
\end{proof}

\begin{lemma}\label{exception}
Let $(R,S)\in\set{P}_2 \cap \set{P}[m_P, \sigma_P, m_Q, \sigma_Q]$, and let $(R',S')\in\set{P}_2 \cap \set{P}[m_P, \sigma_P, m_Q, 0]$.
If $\alpha<-1$, there exixts $\sigma_Q$ such that $d_A^{(\alpha)}(r\|s)>d_A^{(\alpha)}(r'\|s')$.

\end{lemma}
\begin{proof}
See Appendix~\ref{ap_exception}.
\end{proof}

We show Theorem~\ref{theorem_lower_bound} by 4-step.
\begin{proof}[Proof for $0<\alpha<1$] \\
We first prove \eqref{lb_alpha} for pairs of finite discrete probability measures. \\
Let ${D_A^{(\alpha)}}^*:= \inf_{(P,Q)\in\set{P}_n \cap \set{P}[m_P, \sigma_P, m_Q, \sigma_Q]} D_A^{(\alpha)}(P\|Q)$, and suppose ${D_A^{(\alpha)}}^*<d_A^{(\alpha)}(r\|s)$.
By Lemma~\ref{minimum_condition} as $R\rightarrow\infty$ and Lemma~\ref{lem_p2}, there exist sequences of vectors $\{\bm{u}_j\}$, and probability measures $\{P_j\}$ and $\{Q_j\}$, which are defined on $\{\bm{u}_j\}$, such that
\begin{align}
D_A^{(\alpha)}(P_\infty\|Q_\infty)={D_A^{(\alpha)}}^*,
\end{align}
where $Z_\infty$ denotes $\lim_{j\rightarrow \infty} Z_j$ for variables $Z=\{P,Q, u_i\}$. Without any loss of generality, one can assume that $|u_{i,\infty}| < \infty$ for $1\leq i\leq I$ and $|u_{i,\infty}| = \infty$ for $I< i\leq n$. Let $\sum_{i>I} p_{i,\infty}u^2_{i,\infty}=C^2$ and  $\sum_{i> I} q_{i,\infty}u^2_{i,\infty}=D^2$, where $p_{i,j}=P_j(u_{i,j})$ and $q_{i,j}=Q_j(u_{i,j})$. By the variance constraints, we have $C^2\leq m_P^2+\sigma_P^2$ and $D^2\leq m_Q^2+\sigma_Q^2$. Hence, $p_{i,\infty}=O(u_{i,\infty}^{-2})$ and $q_{i,\infty}=O(u_{i,\infty}^{-2})$ hold for $i >I$, and due to $0<\alpha<1$, it follows that
\begin{align}
\label{eq_moment1}
\sum_{i>I} p_{i,\infty}=\sum_{i> I} p_{i,\infty}u_{i,\infty}=0, \\[0.1cm]
\label{eq_moment2}
\sum_{i>I} q_{i,\infty}=\sum_{i> I} q_{i,\infty}u_{i,\infty}=0, \\[0.1cm] 
\label{eq_moment3}
\sum_{i> I} p_{i,\infty}^\alpha q_{i,\infty}^{1-\alpha}=0.
\end{align}
Let $P'$ and $Q'$ be probability measures defined on $\{u_{1,\infty}, u_{2,\infty}, \cdots, u_{I,\infty}\}$, and let $P'(u_{i,\infty})=p_{i,\infty}$, $Q'(u_{i,\infty})=q_{i,\infty}$ for $1\leq i \leq I$.
From \eqref{eq_moment1}-\eqref{eq_moment3}, it follows that 
\begin{align}
(P',Q')&\in \set{P}_{I,R'} \cap \set{P}\Bigl[m_P, \sqrt{\sigma_P^2 -C^2}, m_Q, \sqrt{\sigma_Q^2 -D^2}\Bigr], \\
{D_A^{(\alpha)}}^*&=D_A^{(\alpha)}(P', Q'),
\end{align}
where we set $R' > \max_{i\leq I} |u_{i, \infty}|$. Since the variances of $P'$ and $Q'$ are non-negative, we have $0 \leq C^2\leq \sigma_P^2$, and $0 \leq D^2\leq \sigma_Q^2$. If ${D_A^{(\alpha)}}^*$ is not a global minimum in $\set{P}_{I,R'} \cap \set{P}\Bigl[m_P, \sqrt{\sigma_P^2 -C^2}, m_Q, \sqrt{\sigma_Q^2 -D^2}\Bigr]$, there exists a sequence in $\set{P}_n \cap \set{P}[m_P, \sigma_P, m_Q, \sigma_Q]$ such that $\sum_{i> I} p_{i,\infty}u^2_{i,\infty}=C^2$, $\sum_{i> I} q_{i,\infty}u^2_{i,\infty}=D^2$, which gives a smaller value than ${D_A^{(\alpha)}}^*$. It contradicts that ${D_A^{(\alpha)}}^*$ is global minimum in $\set{P}_{n} \cap \set{P}[m_P, \sigma_P, m_Q, \sigma_Q]$. Hence, by Lemma~\ref{lem_minimum}, it follows that $(P',Q')\in \set{P}_2$, and
\begin{align}
{D_A^{(\alpha)}}^*=D_A^{(\alpha)}(P'\|Q')= d_A^{(\alpha)}(r'\|s'),
\end{align}
where $(r', s')$ with means $(m_P,  m_Q)$ and variances $(\sigma_P^2 -C^2, \sigma_Q^2 -D^2)$.
By Lemma \ref{decreasing}, since $d_A^{(\alpha)}(r'\|s')$ is monotonically decreasing with respect to $\sigma_P$ and $\sigma_Q$, we have ${D_A^{(\alpha)}}^* = d_A^{(\alpha)}(r'\|s') \geq d_A^{(\alpha)}(r\|s)$. This contradicts the assumption of ${D_A^{(\alpha)}}^*<d_A^{(\alpha)}(r\|s)$, thus we obtain ${D_A^{(\alpha)}}^*=d_A^{(\alpha)}(r\|s)$.

Next, we prove \eqref{lb_alpha} for pairs of probability measures in $\set{P}[m_P,\sigma_P, m_Q, \sigma_Q]$. For sufficiently small $\epsilon$, there exists $R$ such that
\begin{align}
\label{approximate1}
|\int_{|x|>R} px^{k} \mathrm{d}\mu(x)|<\epsilon, \quad |\int_{|x|>R} qx^{k} \mathrm{d}\mu(x)|<\epsilon, \quad \mbox{for} \hspace*{0.15cm}  k=0,1,2.
\end{align}
Since $p^\alpha q^{1-\alpha}\leq p+q$, we have 
\begin{align}
\label{approximate2}
\int_{|x|>R}p^\alpha q^{1-\alpha} \mathrm{d}\mu(x) < 2\epsilon.
\end{align}
In the interval $[-R, R]$, one can approximate probability measures by finite discrete probability measures $(P_d, Q_d)\in\set{P}_{n, R}$ as follows.
\begin{align}
\label{approximate3}
&|\int_{|x|\leq R} px^{k} \mathrm{d}\mu(x)- \sum_i p_i u_i^{k}|<\epsilon, \quad \mbox{for} \hspace*{0.15cm}  k=0,1,2,  \\[0.1cm]
\label{approximate4}
&|\int_{|x|\leq R} qx^{k} \mathrm{d}\mu(x)- \sum_i q_i u_i^{k}|<\epsilon, \quad \mbox{for} \hspace*{0.15cm}  k=0,1,2, \\[0.1cm]
\label{approximate5}
&|\frac{1}{\alpha(\alpha-1)} \int_{|x|\leq R}  p^\alpha q^{1-\alpha} \mathrm{d}\mu(x)- D_A^{(\alpha)}(P_d\|Q_d)|<\epsilon.
\end{align}
From \eqref{approximate2} and \eqref{approximate5}, we have $D_A^{(\alpha)}(P\|Q)=D_A^{(\alpha)}(P_d\|Q_d)+O(\epsilon)$. From \eqref{approximate1}, \eqref{approximate3}, and \eqref{approximate4}, it follows that differences of means and variances between $(P,Q)$ and $(P_d, Q_d)$ are $O(\epsilon)$. 
By applying $D_A^{(\alpha)}(P_d, Q_d)\geq d_A^{(\alpha)}(r_d, s_d)$, it follows that
\begin{align}
D_A^{(\alpha)}(P\|Q)=D_A^{(\alpha)}(P_d, Q_d)+O(\epsilon)\geq d_A^{(\alpha)}(r_d,\|s_d)+O(\epsilon)=d_A^{(\alpha)}(r\|s)+O(\epsilon),
\end{align}
where $(r_d, s_d)\in\set{P}_2 \cap \set{P}[m_{P_d}, \sigma_{P_d}, m_{Q_d}, \sigma_{Q_d}]$, and we use differentiability of $d_A^{(\alpha)}(r\|s)$ with respect to means and variances.
Since one can choose arbitrary small $\epsilon$, we obtain $D_A^{(\alpha)}(P\|Q) \geq d_A^{(\alpha)}(r\|s)$.
\end{proof}

\begin{proof}[Proof for $-1\leq \alpha \leq 0$ and $1\leq\alpha\leq 2$] \\  
Let $(P, Q)\in \set{P}[m_P, \sigma_P, m_Q, \sigma_Q]$, and $(R, S)\in \set{P}_2 \cap \set{P}[m_P, \sigma_P, m_Q, \sigma_Q]$. For all $t\in[0,1]$, probability measures $Q_t=(Q-P)t+P$ and $S_t=(S-R)t+R$ have the same means and variances.
Hence, by the tight lower bounds for $0<\alpha+1< 1$, the relation \eqref{relation_alpha_3} for $-1<\alpha< 0$, and setting $t=1$ in \eqref{relation_alpha_3}, we obtain
\begin{align}
D_A^{(\alpha)}(P\|Q)=(\alpha+1)\int_0^1 t^{\alpha-2}D_A^{(\alpha+1)}(P\|Q_t)\mathrm{d}t \geq (\alpha+1)\int_0^1 t^{\alpha-2}d_A^{(\alpha+1)}(r\|s_t)\mathrm{d}t=d_A^{(\alpha)}(r\|s). \nonumber
\end{align}
The inequality for $1<\alpha<2$ follows due to the duality \eqref{eq-duality}.
Next, we prove for $\alpha=-1, 2$.
By the Hammersley–Chapman–Robbins bound~\cite{chapman1951minimum}, we obtain
\begin{align}
D_A^{(2)}(P\|Q)\geq \frac{a^2}{2\sigma_Q^2}. \nonumber
\end{align}
(see~\cite{nishiyama2020relations}[(180)]).
Since $s(1-s)=\frac{\sigma_Q^2}{4v^2}$ and $r-s=\frac{a}{2v}$ hold due to \eqref{r} and \eqref{s}, it follows that
\begin{align}
d_A^{(2)}(r\|s)=\frac 12\Bigl(\frac{(r-s)^2}{s}+\frac{(1-r-(1-s))^2}{1-s}\Bigr)=\frac{(r-s)^2}{2s(1-s)}=\frac{a^2}{2\sigma_Q^2}. \nonumber
\end{align}
By the duality \eqref{eq-duality}, we obtain inequalities for $\alpha=-1$.
The relation \eqref{relation_alpha_3} for $\alpha=1$ and the duality \eqref{eq-duality} yield inequalities for $\alpha=0, 1$ (see~\cite{nishiyama2020relations}[Theorem 2]).
By combining these results, we obtain lower bounds for $-1\leq \alpha \leq 0$ and $1\leq\alpha\leq 2$.
\end{proof}

\begin{proof}[Proof of the necessary condition] \\ 
Due to the duality $\eqref{eq-duality}$, it is enough to show for $\alpha < -1$. We show an example of $(P, Q)\in \set{P}_3 \cap \set{P}[m_P, \sigma_P, m_Q, \sigma_Q]$ such that $d_A^{(\alpha)}(r\|s) > D_A^{(\alpha)}(P\|Q)$. Let $(R', S')\in \set{P}_2 \cap \set{P}[m_P, \sigma_P, m_Q, 0]$, which is defined on $(u'_1, u'_2)$. We obtain  $s'=1$ due to $\sigma_Q=0$.
Let $(P, Q)\in \set{P}_3$ be 
\begin{align}
(P(u'_1), P(u'_2), P(u'_3))&=\Bigl(r'-\frac{1}{{u'_3}^{2+\delta}},\; 1-r', \;\frac{1}{{u'_3}^{2+\delta}}\Bigr), \nonumber \\ 
(Q(u'_1), Q(u'_2), Q(u'_3))&=\Bigl(1-\frac{\sigma_Q^2}{{u'_3}^{2}},\; 0,\; \frac{\sigma_Q^2}{{u'_3}^{2}}\Bigr),\nonumber
\end{align}
where $\delta$ is a small positive real number such that $2+\alpha\delta>0$.
As $u_3'\rightarrow\infty$, $P$ and $Q$ have means $(m_P, m_Q)$ and variances $(\sigma_P^2, \sigma_Q^2)$.
Since $P(u'_3)^\alpha Q(u'_3)^{1-\alpha}=O({u'_3}^{-(2+\alpha\delta)})$, it follows that $D_A^{(\alpha)}(P\|Q)\rightarrow d_A^{(\alpha)}(r'\|s')$. By Lemma~\ref{exception}, there exists $\sigma_Q$ such that $d_A^{(\alpha)}(r\|s)>D_A^{(\alpha)}(P\|Q)$.
\end{proof}
\begin{proof}[Proof of Item (c)] \\
Since the proof is similar to~\cite{nishiyama2020relations}[Theorem 2], we outline of the proof.  
We construct sequence of probability measures $\{(P_j, Q_j)\}$ with zero mean and respective variances $(\sigma_P^2, \sigma_Q^2)$ for which $D_A^{(\alpha)}(P_j\|Q_j)\rightarrow 0$ as $j\rightarrow \infty$
(without any loss of generality, one can assume that the equal means are equal to zero).
We start by assuming $\min\{\sigma_P^2, \sigma_Q^2\}\geq 1$.
Let 
\begin{align}
\mu_j:= \sqrt{1 + j(\sigma_Q^2-1)}, \nonumber
\end{align} 
and define a sequence of quaternary real-valued random variables with probability mass functions
\begin{align}
Q_j(u) := 
\begin{dcases}
\frac12-\frac{1}{2j}, & \quad u = \pm 1, \nonumber\\
\frac{1}{2j},       & \quad u = \pm \mu_j. \nonumber
\end{dcases}
\end{align}
It can be verified that, for all $j\in \naturals$, $Q_j$ has zero mean and variance $\sigma_Q^2$.

Furthermore, let
\begin{align}
P_j(u) := 
\begin{dcases}
\frac{1}{2}-\frac{\xi}{2j}, & \quad u= \pm 1, \nonumber \\
\frac{\xi}{2j},       & \quad u = \pm \mu_j, \nonumber
\end{dcases}
\end{align}
with
\begin{align}
\xi:= \frac{\sigma_P^2-1}{\sigma_Q^2-1}. \nonumber
\end{align}
If $\xi > 1$, for $j=1, \cdots, \lceil \xi \rceil$, we choose $P_j$ arbitrary with mean $0$ and variance $\sigma_P^2$.
Then,
\begin{align}
D_A^{(\alpha)}(P_j\|Q_j)=d_A^{(\alpha)}\Bigl(\frac{\xi}{j}\|\frac{1}{j}\Bigr)\rightarrow 0. \nonumber
\end{align}
Next, suppose $\min\{\sigma_P^2, \sigma_Q^2\}:= \sigma^2<1$, then construct $P'_j$ and $Q'_j$ as before with variances $\frac{2\sigma_P^2}{\sigma^2}>1$ and $\frac{2\sigma_Q^2}{\sigma^2}>1$, respectively. If $P_j$ and $Q_j$ denote the random variables $P'_j$ and $Q'_j$ scaled by a factor of $\frac{\sigma}{\sqrt{2}}$, then their variances are $\sigma_P^2, \sigma_Q^2$, respectively, and $D_A^{(\alpha)}(P_j,Q_j)=D_A^{(\alpha)}(P'_j,Q'_j)\rightarrow 0$ as we let $j\rightarrow \infty$.
\end{proof}
\section{Conclusion} 
In this paper, we derived relations between $\alpha$ and $\alpha+1$ for the asymmetric $\alpha $-divergences. These relations are generalizations of the integral relation between the Kullback-Leibler divergence and the $\chi^2$-divergence. We showed that $\alpha\in[-1,2]$ is the necessary and sufficient condition that the binary $\alpha$-divergences always attain lower bounds under given means and variances. Kullback-Leibler divergence, Hellinger distance, and $\chi^2$-divergence satisfy this condition. It is intuitively natural that discrepancy between probability measures under moment constraints is smaller for localized measures than for broad probability measures.
In this point of view, $\alpha$-divergences for $\alpha\in[-1,2]$ have preferable properties. 
The tight lower bounds for the Kullback-Leibler divergence and the Hellinger distance recently began to be applied to physics~\cite{hasegawa2021irreversibility, van2020unified}.
In the future, we hope that the range of applications of relations between the $\alpha$-divergences and tight lower bounds will expand, and we hope that they will help to progress many fields and to deepen the understanding of properties of divergences.

\bibliography{reference_Alpha_div} 
\bibliographystyle{myplain} 

\appendices
\section{Proof of Lemma~\ref{minimum_condition}}
\label{ap_minimum_condition}
\begin{proof}
Let $n\geq 3$ ($n=1,2$ are trivial), and let $p_i:= P(u_i)$, $q_i:= Q(u_i)$. 
Consider the following minimization problem.
\begin{align}
\label{eq_min_problem}
&\mbox{minimize} \quad - \sum_i p_i^\alpha q_i^{1-\alpha}, \\[0.1cm]
\label{eq_constraint1}
\quad \mbox{subject to}\quad &g_k(\bm{p}, \bm{u}):= \sum_i p_i u_i^{k-1} -A_k=0, \\[0.1cm] 
\label{eq_constraint2}
& g_{k+3}(\bm{q}, \bm{u}):= \sum_i q_i u_i^{k-1} -B_k=0, \quad \mbox{for} \hspace*{0.15cm}  k=1,2,3,\\[0.1cm]
\label{eq_constraint3}
&0\leq p_i \leq 1, \quad 0\leq q_i \leq 1, \quad |u_i|\leq R, \quad \mbox{for} \hspace*{0.15cm}  1\leq i\leq n,
\end{align}
where $\bm{A}:= (1, m_P, \sigma_P^2 + m_P^2)^\mathrm{T}$, $\bm{B}:= (1, m_Q, \sigma_Q^2 + m_Q^2)^\mathrm{T}$, and \eqref{eq_constraint1} and \eqref{eq_constraint2} correspond to \eqref{constraints}.
Since the feasible set is compact, there exists a global minimum.
We first consider $\sigma_P > 0$ and $\sigma_Q > 0$, then we have $p_i ,q_i< 1$ for $1\leq i\leq n$.
Hence, the global minimum point must be a stationary point, or be on the boundary at $\max_i |u^*_i| = R$, $p^*_i=0$, or $q^*_i=0$. By rearranging the order of $\{p^*_i\}, \{q^*_i\}$ appropriately, let
\begin{align}
p^*_i>0,  \; q^*_i &> 0, \quad \mbox{for} \hspace*{0.15cm} i\leq I, \nonumber \\[0.1cm]
p^*_j>0,  \; q^*_j &= 0, \quad \mbox{for} \hspace*{0.15cm} I < j\leq I+J,  \nonumber \\[0.1cm]
p^*_k=0,  \; q^*_k &> 0, \quad \mbox{for} \hspace*{0.15cm} I+J < k\leq I+J+K=n. \nonumber
\end{align}
Since $- \sum_{i\leq I} p_i^\alpha q_i^{1-\alpha}<0$ holds for the global minimum point, we have $I\geq 1$.
If $J>0$, let $(P', Q', \bm{u'})$ be 
\begin{align}
&p'_i=p^*_i+\dif p_i,  \quad q'_i=q^*_i+\dif q_i, \quad \mbox{for} \hspace*{0.15cm} 1\leq i\leq I, \nonumber\\[0.1cm]
&p'_{I+1}=p^*_{I+1}+\dif p_{I+1}, \quad q'_{I+1}=\epsilon >0, \nonumber\\[0.1cm]
&p'_j=p^*_j+\dif p_j,  \quad q'_j=0, \quad \mbox{for} \hspace*{0.15cm} I+1 < j\leq I+J, \nonumber\\[0.1cm]
&p'_k=0,  \quad q'_k=q^*_k+\dif q_k,  \quad \mbox{for} \hspace*{0.15cm} I+J < k\leq I+J+K=n, \nonumber\\[0.1cm]
&u'_i=u^*_i+\dif u_i,  \quad \mbox{for} \hspace*{0.15cm} i\leq n, \nonumber
\end{align}
where $\epsilon$ and $\dif \cdot$ denote sufficiently small real numbers.
The moment constraints for probability measures $P'$ and $Q'$ include $\{\dif p_1, \dif p_{I+1}, \dif u_{I+1}\}$ and $\{\dif q_1, \dif q_{l}, \dif u_{l}\}$ ($l\neq I+1$), respectively.
These variables are independent, and it can be easily verified that the determinants for each probability measure constraint are non-zero by $u^*_i\neq u^*_j$ for $i\neq j$. Hence, one can choose $\dif \cdot$ such that $(P', Q')\in \set{P}[m_P, \sigma_P, m_Q, \sigma_Q]$ and $\dif \cdot = O(\epsilon)$.
By $0<\alpha<1$ and $p^*_{I+1}>0$, 
\begin{align}
-\sum_{i\leq I} {p^*}_i^\alpha {q^*}_i^{1-\alpha}-\Bigl(-\sum_{i\leq I} {p'}_i^\alpha {q'}_i^{1-\alpha}-{p'}_{I+1}^\alpha {q'}_{I+1}^{1-\alpha}\Bigr)= \epsilon^{1-\alpha} {p^*_{I+1}}^{\alpha} + O(\epsilon) > 0.
\end{align} 
This contradicts that $(\bm{p}^*, \bm{q}^*, \bm{u}^*)$ is a global minimum, then we obtain $J=0$.
In a similar way, we also obtain $K=0$ and $I=n$. Hence, the global minimum point is an interior point or on the the boundary at $\max_i |u^*_i| = R$.
Supposing $\max_i |u^*_i| < R$, it must be a stationary point of the following Lagrangian.
\begin{align}
L(\bm{p}, \bm{q}, \bm{u}, \boldsymbol{\lambda}):=- \sum_{i\leq n} p_i^\alpha q_i^{1-\alpha}+ \sum_{i\leq n} p_i\phi_{\boldsymbol{\lambda}}(u_i) + \sum_{i\leq n} q_i\psi_{\boldsymbol{\lambda}}(u_i)-\sum_{k=1}^3 \lambda_k A_k-\sum_{k=1}^3 \lambda_{k+3}B_k,
\end{align}
where $\phi_{\boldsymbol{\lambda}}(u):= \sum_{k=1}^3 \lambda_k u^{k-1}$ and  $\psi_{\boldsymbol{\lambda}}(u):= \sum_{k=1}^3 \lambda_{k+3} u^{k-1}$.
Since $u_i\neq u_j$ for $i\neq j$, and $n\geq 3$, it follows that $\{\nabla g_k\}_{k\leq 6}$ are linearly independent. 
The stationary conditions are 

\begin{align}
\label{grad_1}
\frac{\partial{L}}{\partial{p_i}}(\bm{p}^*, \bm{q}^*, \bm{u}^*, \boldsymbol{\lambda}^*)&= -\alpha{\Bigl(\frac{p^*_i}{q^*_i}\Bigr)}^{\alpha-1}+ \phi_{\boldsymbol{\lambda}^*}(u^*_i)=0, \\[0.1cm]
\label{grad_2}
\frac{\partial{L}}{\partial{q_i}}(\bm{p}^*, \bm{q}^*, \bm{u}^*, \boldsymbol{\lambda}^*)&= -(1-\alpha){\Bigl(\frac{p^*_i}{q^*_i}\Bigr)}^\alpha+ \psi_{\boldsymbol{\lambda}^*}(u^*_i)=0, \\[0.1cm]
\label{grad_3}
\frac{\partial{L}}{\partial{u_i}}(\bm{p}^*, \bm{q}^*, \bm{u}^*, \boldsymbol{\lambda}^*)&= p^*_i\phi'_{\boldsymbol{\lambda}^*}(u^*_i) +  q^*_i\psi'_{\boldsymbol{\lambda}^*}(u^*_i)=0,
\end{align}
where $'$ denotes the derivative with respect to $u$.

From \eqref{grad_1} and \eqref{grad_2}, it follows that
\begin{align}
\label{eq_condition1}
\psi_{\boldsymbol{\lambda}^*}(u^*_i)=(1-\alpha){\Bigl(\frac{\phi_{\boldsymbol{\lambda}^*}(u^*_i)}{\alpha}\Bigr)}^{\frac{\alpha}{\alpha-1}}.
\end{align}
Substituting \eqref{grad_1} and \eqref{grad_3} into \eqref{grad_2}, we have
\begin{align}
\label{eq_condition2}
\alpha\phi'_{\boldsymbol{\lambda}^*}(u^*_i)\psi_{\boldsymbol{\lambda}^*}(u^*_i)+(1-\alpha)\phi_{\boldsymbol{\lambda}^*}(u^*_i)\psi'_{\boldsymbol{\lambda}^*}(u^*_i)=0.
\end{align}
Since $\phi_{\boldsymbol{\lambda}^*}(u^*)$ and $\psi_{\boldsymbol{\lambda}^*}(u^*)$ are positive from \eqref{grad_1} and \eqref{grad_2}, one can define the function $\beta(u):=\alpha\log\frac{\phi_{\boldsymbol{\lambda}^*}(u)}{\alpha}+(1-\alpha)\log\frac{\psi_{\boldsymbol{\lambda}^*}(u)}{1-\alpha}$.
From  \eqref{eq_condition1} and \eqref{eq_condition2}, we obtain
\begin{align}
\label{eq_condition3}
\beta(u^*_i)=\beta'(u^*_i)=0.
\end{align}
Since $\phi_{\boldsymbol{\lambda}^*}(u)$ and $\psi_{\boldsymbol{\lambda}^*}(u)$ are at most quadratic functions with respect to $u$, the equation \eqref{eq_condition2} has at most a degree of $3$. If \eqref{eq_condition2} is not identically $0$, we have $n\leq 3$.
If $n=3$, by the relation \eqref{eq_condition3} and the mean value theorem, there exists $u_l \neq u^*_1, u^*_2, u^*_3$ such that $\beta'(u_l)=0$.
By the definition of $\beta(\cdot)$, it follows that $u_l$ is a solution of \eqref{eq_condition2}.
This contradicts that the equation \eqref{eq_condition2} has at most a degree of $3$.
If equation \eqref{eq_condition2} is identically $0$, the equality \eqref{eq_condition1} is identity. Since the degree of $\phi_{\boldsymbol{\lambda}^*}(u)$ and $\psi_{\boldsymbol{\lambda}^*}(u)$ are any of $\{0,\;1, \;2\}$, and due to $0<\alpha<1$, it follows that they are constant.
The equality \eqref{grad_1} yields $q_i = Cp_i$ for all $i$ and a constant $C$. This implies $p_i=q_i$ for all $i$, and it contradicts $m_P\neq m_Q$.
Hence, we obtain $n\leq 2$.

Next, we show the case for $\sigma_P>0$ and $\sigma_Q=0$ briefly. The notation is the same as above, and $I=1$ and $J=n-1$ hold due to $\sigma_Q=0$.
The Lagrangian is
\begin{align}
L(\bm{p}, \bm{u}, \boldsymbol{\lambda}):=-p_1^\alpha+ \sum_{i\leq n} p_i\phi_{\boldsymbol{\lambda}}(u_i) -\sum_{k=1}^3 \lambda_k A_k.
\end{align}
Supposing $\max_i |u^*_i| < R$, the stationary conditions are
\begin{align}
\label{grad_8}
\frac{\partial{L}}{\partial{p_1}}(\bm{p}^*, \bm{u}^*, \boldsymbol{\lambda}^*)&= -\alpha {p^*_1}^{\alpha-1}+ \phi_{\boldsymbol{\lambda}^*}(u^*_1)=0, \\[0.1cm]
\label{grad_9}
\frac{\partial{L}}{\partial{p_i}}(\bm{p}^*, \bm{u}^*, \boldsymbol{\lambda}^*)&= \phi_{\boldsymbol{\lambda}^*}(u^*_i)=0, \quad \mbox{for} \hspace*{0.15cm}  2\leq i\leq n,\\[0.1cm]
\label{grad_10}
\frac{\partial{L}}{\partial{u_i}}(\bm{p}^*, \bm{u}^*, \boldsymbol{\lambda}^*)&= p^*_i\phi'_{\boldsymbol{\lambda}^*}(u^*_i) =0, \quad \mbox{for} \hspace*{0.15cm}  1\leq i\leq n.
\end{align}
From \eqref{grad_8} and \eqref{grad_9}, it follows that $\phi_{\boldsymbol{\lambda}^*}(u^*_i)$ is not a constant. Since $\phi'_{\boldsymbol{\lambda}^*}(u)$ is at most a dgree of 1 from \eqref{grad_10} and $p^*_i>0$ for $2\leq i\leq n$, it follows that $n\leq 2$. The result for $\sigma_P=0$ and $\sigma_Q>0$ also follows by swapping $P$ and $Q$.
\end{proof}

\section{Proof of Lemma~\ref{decreasing}}
\label{ap_decreasing}
We first prove the following lemma.
\begin{lemma} \label{alpha_inequality}
Let $0< \alpha<1$.
If $0< x \leq 1$, 
\begin{align}
\frac{(1-\alpha x)(1+x)^\alpha}{(1+\alpha x)(1-x)^\alpha}> 1. \nonumber
\end{align}
If $-1\leq x <0$, 
\begin{align}
\frac{(1-\alpha x)(1+x)^\alpha}{(1+\alpha x)(1-x)^\alpha}< 1.  \nonumber
\end{align}
\end{lemma}

\begin{proof}
Letting $F(x):= \log\frac{(1-\alpha x)(1+x)^\alpha}{(1+\alpha x)(1-x)^\alpha} $, for $0<|x|\leq 1$ and $0<\alpha < 1$, we have
\begin{align}
F'(x)=\alpha\Bigl(\frac{1}{1+x}+\frac{1}{1-x}-\frac{1}{1-\alpha x}-\frac{1}{1+\alpha x}\Bigr)= \frac{2\alpha (1-\alpha^2)x^2}{(1-x^2)(1-\alpha^2x^2)}> 0. \nonumber
\end{align}
By combining this relation with $F(0)=0$, the results follow.
\end{proof}

\begin{proof} [Proof of Lemma~\ref{decreasing}]
Let $V_P:= \sigma_P^2$ and $V_Q:= \sigma_Q^2$, and $f^{(\alpha)}(V_P,V_Q,a):= r^\alpha s^{1-\alpha}$.
From \eqref{r}-\eqref{v}, by replacing $a$ by $-a$, we obtain $(r,s)\rightarrow (1-r, 1-s)$.
Hence, the binary $\alpha$-divergences are written by $d_A^{(\alpha)}(r\|s)=\frac{1}{\alpha(\alpha-1)}(f^{(\alpha)}(V_P,V_Q,a)+f^{(\alpha)}(V_P,V_Q,-a)-1)$.
Since $0<\alpha<1$, we show that $f^{(\alpha)}(V_P,V_Q,a)+f^{(\alpha)}(V_P,V_Q,-a)$ is monotonically increasing with respect to $V_P$ and $V_Q$.

From \eqref{r}-\eqref{v}, we have
\begin{align}
\label{dif_v}
\frac{\partial v}{\partial V_Q}&=\frac{V_Q-V_P+a^2}{4a^2v}, \\[0.1cm]
\label{dif_r}
\frac{\partial r}{\partial V_Q}&=-\frac{(V_Q-V_P+a^2)^2}{16a^3v^3}+\frac{1}{4av}=\frac{V_P}{4av^3}, \\[0.1cm]
\label{dif_s}
\frac{\partial s}{\partial V_Q}&=-\frac{(V_Q-V_P+a^2)(V_Q-V_P-a^2)}{16a^3v^3}+\frac{1}{4av}=\frac{V_P+V_Q+a^2}{8av^3}.
\end{align}
By combining \eqref{dif_r} and \eqref{dif_s}, it follows that
\begin{align}
\label{dif_f}
\frac{\partial f^{(\alpha)}}{\partial V_Q}=\frac{1}{8av^3}{\Bigl(\frac{r}{s}\Bigr)}^\alpha\Bigr((1-\alpha)(V_P+V_Q+a^2)+2\alpha V_P\frac{s}{r}\Bigr).
\end{align}
By $r-s=\frac{a}{2v}$ and $s(1-s)=\frac{V_Q}{4v^2}$, we obtain
\begin{align}
\label{rs}
\frac{r}{s}=1+\frac{a}{2vs}=1+\frac{2va(1-s)}{V_Q}=\frac{V_P+V_Q+a^2+2va}{2V_Q}.
\end{align}
By replacing $(V_P, V_Q, a)$ by $(V_Q, V_P, -a)$, we obtain $(r,s)\rightarrow (s,r)$.
From \eqref{rs}, we have 
\begin{align}
\label{sr}
\frac{s}{r}=\frac{V_P+V_Q+a^2-2va}{2V_P}.
\end{align}
Substituting \eqref{rs} and \eqref{sr} into \eqref{dif_f}, we obtain
\begin{align}
\label{dif_f2}
\frac{\partial f^{(\alpha)}}{\partial V_Q}(V_P,V_Q,a)&=\frac{1}{8av^3}{\Bigl(\frac{V_P+V_Q+a^2+2va}{2V_Q}\Bigr)}^\alpha\Bigr((V_P+V_Q+a^2)-2\alpha va\Bigr) \\
&=\frac{1}{8av^3}{\Bigl(\frac{V_P+V_Q+a^2}{2V_Q}\Bigr)}^\alpha(V_P+V_Q+a^2)(1+x)^\alpha(1-\alpha x),
\end{align}
where $x:=\frac{2va}{V_P+V_Q+a^2}$ and $|x|\leq 1$.
Hence, we obtain 
\begin{align}
\label{dif_f3}
&\frac{\partial\Bigl(f^{(\alpha)}(V_P,V_Q,a)+f^{(\alpha)}(V_P,V_Q,-a)\Bigr)}{\partial V_Q} \nonumber \\
&=\frac{1}{8av^3}{\Bigl(\frac{V_P+V_Q+a^2}{2V_Q}\Bigr)}^\alpha(V_P+V_Q+a^2)\Bigl((1+x)^\alpha(1-\alpha x)-(1-x)^\alpha(1+\alpha x)\Bigr).
\end{align}
If $a>0$, by Lemma~\ref{alpha_inequality}, \eqref{dif_f3} and $0< x \leq 1$, it follows that 
\begin{align}
\label{dif_f4}
\frac{\partial\Bigl(f^{(\alpha)}(V_P,V_Q,a)+f^{(\alpha)}(V_P,V_Q,-a)\Bigr)}{\partial V_Q}> 0.
\end{align}

If $a<0$, we obtain \eqref{dif_f4} due to $-1\leq x<0$ and Lemma~\ref{alpha_inequality}.
By \eqref{r}, \eqref{s}, and the definition of $f^{(\alpha)}(V_P,V_Q,a)$, we have $f^{(\alpha)}(V_P,V_Q,a)= f^{(1-\alpha)}(V_Q,V_P,-a)$.
Hence, by combinig this relation with \eqref{dif_f4} for $1-\alpha$, we have $\frac{\partial\Bigl(f^{(\alpha)}(V_P,V_Q,a)+f^{(\alpha)}(V_P,V_Q,-a)\Bigr)}{\partial V_P}> 0$.
\end{proof}

\section{Proof of Lemma~\ref{exception}}
\begin{proof}
\label{ap_exception}
Let $\alpha < -1$.
In a similar way to the proof of Lemma~\ref{decreasing}, we obtain \eqref{dif_f3}.
For $a>0$, let 
\begin{align}
x(V_Q)=\frac{2va}{V_P+V_Q+a^2}, \nonumber
\end{align}
where we write $V_Q$ explicitly in $x$ in the proof of Lemma~\ref{decreasing}.
By $x(0)=1$, there exists $V_Q=\sigma_Q^2$ such that $x(z)>0$ and $1+\alpha x(z) < 0$ for all  $z\in[0, V_Q]$. 
By combining these inequalities with $|x(V_Q)|\leq 1$, it follows that \eqref{dif_f3} is positive for all $z\in[0, V_Q]$. By $d_A^{(\alpha)}(r\|s)=\frac{1}{\alpha(\alpha-1)}(f^{(\alpha)}(V_P,V_Q,a)+f^{(\alpha)}(V_P,V_Q,-a)-1)$ and $\alpha(\alpha-1)>0$, it follows that
\begin{align}
d_A^{(\alpha)}(r\|s) > d_A^{(\alpha)}(r'\|s'), 
\end{align}
where $(R', S')\in \set{P}_2 \cap \set{P}[m_P, \sigma_P, m_Q, 0]$. The case for $a<0$ can be justified in a similar way.

\end{proof}
\end{document}